\documentclass[12pt]{article}

\usepackage{amssymb,amsthm,amsfonts,amsmath,amscd}
\usepackage{fullpage}
\usepackage{float}

\theoremstyle{plain}
\newtheorem{theorem}{Theorem}

\theoremstyle{definition}

\newtheorem{example}[theorem]{Example}

\theoremstyle{remark}
\newtheorem{remark}[theorem]{Remark}

\DeclareMathOperator{\diag}{diag}

\begin{document}

\author{Hamoon Mousavi and Jeffrey Shallit\\
School of Computer Science \\
University of Waterloo\\
Waterloo, ON  N2L 3G1 \\
Canada \\
{\tt hamoon.mousavihaji@uwaterloo.ca} \\
{\tt shallit@cs.uwaterloo.ca}\\
}

\title{Filtrations of Formal Languages by Arithmetic Progressions}

\maketitle

\begin{abstract}
A filtration of a formal language $L$ by a sequence
$s$ maps $L$ to the set of words
formed by taking the letters of words of $L$ indexed only by $s$.
We consider the languages resulting
from filtering by all arithmetic progressions.  If $L$ is regular,
it is easy to see 
that only finitely many distinct languages result.  By
contrast, there exist CFL's that give infinitely
many distinct languages as a result.  We use our technique to
show that the operation $\diag$, which extracts
the diagonal of words of square length arranged in a square array,
preserves regularity but does not preserve context-freeness.
\end{abstract}

\section{Introduction}

Let $s = (s(i))_{i \geq 0}$ be an infinite strictly increasing sequence
of non-negative integers.  Berstel et al.~\cite{Berstel&Boasson&Carton&Petazzoni&Pin:2006} introduced the notion of filtering by $s$:  given a finite word
$w = a_0 a_1 \cdots a_n$, we write $w[s] = a_{s(0)} a_{s(1)} \cdots
a_{s(k)}$, where $k$ is the largest integer such that $s(k) \leq n < s(k+1)$.
(If there is no such integer, then $w[s] = \epsilon$.)
Given a language $L$, we define $L[s] = \lbrace w[s] \ : \ w \in L \rbrace$.

\begin{example}
If $w = {\tt theorem}$, and $s = 0,2,4,6,\ldots$, the
sequence of even integers, then $w[s] = {\tt term}$.  
If $t = 1,3,5,\ldots$, the sequence of odd integers, then
$w[t] = {\tt hoe}$.
\end{example}

Berstel et al. \cite{Berstel&Boasson&Carton&Petazzoni&Pin:2006}
proved a number of theorems about filters, and characterized
those sequences $s$ that preserve regularity (i.e., $L[s]$ is always regular
if $L$ is) and context-freeness.

In this note we revisit the concept of filtering from a slightly different
point of view.    Suppose we have an infinite set of filters 
$S = \lbrace s_1, s_2, \ldots \rbrace$.  Given a language $L$, what
can be said about the set of all filtered languages
$\lbrace L[s_i] \ : \ i \geq 1 \rbrace$?   For example, is it finite?

In this note we are only concerned with filters $s$ that represent
{\it arithmetic progressions}:  there exist integers $a \geq 1, b \geq 0$
such that $s_i = ai + b$ for $i \geq 0$.  We consider four different
types of filter sets:
\begin{itemize}
\item[(a)] $a \geq 1$ and $b = 0$:  the weak arithmetic progressions
\item[(b)] $a \geq 1$ and $0 \leq b < a$:  the ordinary arithmetic progressions
\item[(c)] $a \geq 1$ and $b \geq 0$:  the strong arithmetic progressions
\item[(d)] $a = 1$ and $b \geq 0$:  the shifts
\end{itemize}

If $L$ is regular, a simple argument (given below) shows that filtration by
the strong arithmetic progressions produces only finitely many distinct
languages (and hence the same is true for filtration by the weak and
ordinary arithmetic progressions and shifts).
By contrast, there exist context-free languages $L$ so that
filtering only by the weak arithmetic progressions or the shifts produces
infinitely many distinct languages (and hence the same is true
for the ordinary and strong arithmetic progressions).

In Section~\ref{diag-sec}
we introduce a natural operation on formal languages that is related
to the results of 
Berstel et al.~\cite{Berstel&Boasson&Carton&Petazzoni&Pin:2006},
but seemingly cannot be analyzed using their framework.
We show that this operation
preserves regularity, but does not preserve context-freeness.

We adopt the following notation:  if $L$ is a language,
and $s = (s_i)_{i \geq 0}$ is an
arithmetic progression such that $s_i = ai + b$, then we 
define $L_{a,b} := L[s]$.  Similarly, if $w$ is a word, we define
$w_{a,b} := w[s]$.  

\section{The regular case}

\begin{theorem}
If $L$ is regular, then filtering by the strong arithmetic progressions
produces finitely many distinct languages.
\label{reg1}
\end{theorem}

\begin{remark}
It is easy to see that if $L$ is regular and
$s$ is an arithmetic progression, then $L[s]$ is regular.
Indeed, this follows immediately from the theorem that the regular
languages are closed under applying a transducer, since it is easy
to make a transducer that extracts the letters corresponding to
indices in $s$.  That is not the issue here; we need to see that
among all the regular languages produced by filtering by a
strong arithmetic progression, there are only finitely many distinct
languages.
\end{remark}

\begin{proof}
Let $A = (Q, \Sigma, \delta, q_0, F)$ be a DFA accepting $L$.  Our
proof is based on the boolean matrix interpretation of automata
\cite{Zhang:1999}.
Let $M_c$ be the boolean incidence matrix of the underlying transition graph of
the automaton corresponding to a transition on the symbol $c \in \Sigma$.
That is, if $Q = \lbrace q_0, q_1, \ldots, q_{n-1} \rbrace$, then 
\begin{displaymath}
(M_c)_{i,j} = \begin{cases}
	1, & \text{if $\delta(q_i, c) = q_j$}; \\
	0, & \text{otherwise}.
	\end{cases}
\end{displaymath}
We also write $M = \bigvee_{c \in \Sigma} M_c$.  
By standard results about path algebra,
the matrix $M^n$ has a $1$ in row $i$ and column $j$ if and only if
there is a length-$n$ path from $q_i$ to $q_j$.

Suppose $L = L(A)$.  We show how to create a DFA
$A= (Q', \Sigma, \delta', q'_0, F')$ accepting
$L_{a,b}$.     The idea is that $w = c_0 c_1 \cdots c_{n-1}$ 
should be accepted if and only if there exists a word
$x \in L$ such that
$$x = x_0 c_0 x_1 c_1 \cdots x_{n-1} c_{n-1} x_n,$$
where $x_0, x_1, \ldots, x_n$ are words such that
$|x_0| = b$, $|x_i| = a-1$ for $1 \leq i < n$, and
$|x_n| < a$.

The state set is
$Q' = \lbrace q'_0 \rbrace \ \cup \ \lbrace 0, 1 \rbrace^n$.
Thus all states except $q'_0$ are boolean vectors.
We let $f$ be a boolean vector with $1$'s in the positions
corresponding to final states of $F$.

We define the transition function $\delta'$ as follows:
\begin{eqnarray*}
\delta'(q'_0, c) &=& [1 \ \overbrace{0 \ 0 \cdots \ 0}^{n-1} \, ]\, M^b M_c; \\
\delta'(q, c) &=& q M^{a-1} M_c,
\end{eqnarray*}
for all boolean vectors $q$ and symbols $c \in \Sigma$.
Also define
$$T = \lbrace q \ : \ \text{ there exists  $i$, $0 \leq i < a$,
such that $q \cdot M^i \cdot f = 1$ } \rbrace.$$
Finally, set
\begin{displaymath}
F' = 
\begin{cases}
T \ \cup \ \lbrace q'_0 \rbrace, & \text{ if $L$ contains a word of
length $\leq b$;} \\
T, & \text{otherwise.}
\end{cases}
\end{displaymath}

An easy induction on $n$ now shows that if
$\delta'(q'_0, c_0 c_1 \cdots c_{n-1}) = v$, then
$v$ has $1$'s in the positions corresponding to all states of the
form $\delta(q_0, x_0 c_0 \cdots x_{n-1} c_{n-1})$,
where the words $x_i$ satisfy the inequalities mentioned previously.
It follows that $L(A') = L_{a,b}$.

Note that
$A'$ has $2^n + 1$ states, and this quantity
does not depend on $a$ or $b$.
There are only finitely many
languages with this property.
\end{proof}

\section{The context-free case}

\begin{theorem}  There exists a context-free language
$L$ such that filtering by the weak arithmetic progressions
produces infinitely many distinct languages.
\end{theorem}

\begin{proof}
Consider the
language 
$$L = \lbrace 1 0^n 2 (0^+ 3)^n \ : \ n \geq 1 \rbrace.$$
Then it is easy to see that $L$ is context-free, as it is generated by
the context-free grammar
\begin{eqnarray*}
S & \rightarrow & 10AB \\
A & \rightarrow & 0AB \ | \ 2 \\
B & \rightarrow & 0B \ | \ 03
\end{eqnarray*}
We claim that the languages $L_{a,0}$ for $a \geq 2$ are all distinct.
To see this, it suffices to show that
$L_{a,0} \ \cap \ 123^+  = \lbrace 123^{a-1} \rbrace $.

Clearly $123^{a-1} = z_{a,0}$, where $z = 1 0^{a-1} 2 (0^{a-1} 3)^{a-1}
\in L$.

Now suppose $x \in L_{a,0} \ \cap \ 123^+$.
Then $x = w_{a,0}$ for some $w \in L$.
Since each word in $L$ starts $1 0^n 2$ and contains no other 
$2$'s, we must have $n = a-1$.  It follows that $w \in
1 0^{a-1} 2 (0^+ 3)^{a-1}$.  But then $w$ contains only $a-1$ $3$'s, so
to get $a-1$ $3$'s in $x$, each of them must be used.
It follows that the exponent of $0$ in each $0^+ 3$ is
$a-1$, and so $x = 123^{a-1}$.

This completes the proof.
\end{proof}

\begin{theorem}
There exists a context-free language such that $L$ filtered by 
the shifts results in infinitely many distinct languages.
\end{theorem}

\begin{proof}
Let $L = \lbrace 0^n 1^n \ : \ n \geq 0 \rbrace$.  Then
each of the languages $L_{1,b}$ is distinct, as for each $b \geq 0$, the
word $1^b$ is the longest word of the form $1^*$ in $L_{1,b}$.
\end{proof}

\section{The operation $\diag$}
\label{diag-sec}

Inspired by \cite{Lepisto&Pappalardi&Saari:2007}, which considered the
transposition of words arranged into square arrays,
we introduce the
following natural operation on words of length $n^2$ for some
integer $n \geq 1$:  we arrange the letters of 
the word $w = a_0 a_1 \cdots a_{n^2-1}$ in row major order in
a square array,
\begin{displaymath}
\left[
\begin{array}{cccc}
a_0 & a_1 & \cdots & a_{n-1} \\
a_n & a_{n+1} & \cdots & a_{2n-1} \\
\vdots & \vdots & \ddots & \vdots \\
a_{n^2-n} & a_{n^2-n+1} & \cdots & a_{n^2-1}
\end{array}
\right]
\end{displaymath}
and then take the diagonal $a_0 a_{n+1} a_{2n+2} \cdots a_{n^2-1}$.
We call the result $\diag(w)$.  Thus, for example,
$\diag({\tt absorbent}) = {\tt art} $.  Diagonals of matrices
have long been studied in mathematics.
We extend $\diag$ to languages $L$ as follows:
$$ \diag(L) = \lbrace \diag(w) \ : \ w \in L \text{ and there exists
$n \geq 1$ such that } |w| = n^2 \rbrace.$$

\begin{theorem}
If $L$ is regular then so is $\diag(L)$.
\end{theorem}

\begin{proof}
Given a DFA $A = (Q, \Sigma, \delta, q_0, F)$ accepting $L$,
we construct an NFA $A' = (Q', \Sigma, \delta', q'_0, F')$ accepting 
$\diag(L)$.   As in the proof of Theorem~\ref{reg1},
we let $M_c$ be the $n \times n$ boolean incidence matrix of
the underlying transition graph of
the automaton corresponding to a transition on the symbol $c \in \Sigma$,
and we define $M = \bigvee_{c \in \Sigma} M_c$.

The idea is that $w = a_1 \cdots a_t \in L(A')$ if and only if there exists
$x \in L(A)$ such that $x = a_1 x_1 \cdots a_{t-1} x_{t-1} a_t$
where $|x_i| = t$ for $1 \leq i < t$.    

The states of $A'$ are of the form $[v, V, W]$ where $v$ is 
a length-$n$ boolean vector and $V$ and $W$ are $n \times n$ boolean
arrays.  Let $i = [1 \ \overbrace{0 \ 0 \cdots \ 0}^{n-1} \, ]$
and $f$ be the boolean vector corresponding to the final states of $A$.
The transitions of $A'$ are given by
\begin{eqnarray*}
\delta'(q'_0, c) &=& \lbrace [i \cdot M_c, M, X] \ : \ 
	\exists n\geq 0 \text{ such that } X = M^n \rbrace \\
\delta'([v, V, W], a) &=& \lbrace [v M_a W, VM, W] \rbrace .
\end{eqnarray*}
for all $c \in \Sigma$, and boolean vectors $v$, and boolean
matrices $V, W$.
The final states of $A'$ are 
$$F' = \lbrace [v, V, W] \ : \ vf = 1 \text{ and } V = W \rbrace.$$
We leave it to the reader to verify that $L(A') = \diag(L)$.
\end{proof}

\begin{theorem} 
There exists a context-free language $L$ such that
$\diag(L)$ is not context-free.
\end{theorem}

\begin{proof}
For expository reasons,
our example is over the alphabet
$\lbrace a,b,c,d,e,f,g,h,i,j,0 \rbrace$ of $11$ letters, although it is easy
to reduce this.

Consider
$$ L = \lbrace a 0^{3m+1} b (0^+ c)^{m-2} 0^+ d 0^{3n+1} e (0^+ f)^{n-2} 0^+ g
 0^{3p+1} h (0^+ i)^{p-2} 0^+ j \ : \ m, n, p \geq 3 \rbrace. $$
It is clear that $L$ is context-free, as it is the concatenation 
$L_1 L_2 L_3$ of
the three languages
\begin{eqnarray*}
L_1 &=& \lbrace a 0^{3m+1} b (0^+ c)^{m-2} 0^+ \ : \ m \geq 3 \rbrace \\
L_2 &=& \lbrace d 0^{3n+1} e (0^+ f)^{n-2} 0^+ \ : \ n \geq 3 \rbrace \\
L_3 &=& \lbrace g 0^{3p+1} h (0^+ i)^{p-2} 0^+ j \ : \ p \geq 3 \rbrace 
\end{eqnarray*}
each of which is easily seen to be context-free.

We will show that $\diag(L)$ is not context-free by showing that
$$L' := \diag(L) \ \cap \ a b c^+ d e f^+ g h i^+ j $$
is not context-free.

We claim that $L' = \lbrace a b c^t d e f^t g h i^t j \ : \ t \geq 1 \rbrace$.
It is easy to see that every word of the form
$a b c^t d e f^t g h i^t j$ for $t \geq 1$ is in $L'$, since we can 
take $m = n = p = t+2$,
and the exponent of $0$ in each $0^+$ term to be $3m+1$.  

It remains to see that these are the only words of the form
$a b c^+ d e f^+ g h i^+ j$ in $L'$.
Let $x \in L'$, and let $y \in L$ such that $x = \diag(y)$.  Then
since the first two symbols of $x$ must be $ab$, and since they
are separated by $3m+1$ $0$'s for some $m \geq 3$, it must be that
$|y| = (3m+1)^2$.  Then $|x| = 3m+1$.
We can repeat the argument with the letters
$d, e$ and $g,h$ to get $m = n = p$.  
Removing the single occurrence of each letter $a, b, d, e, g, h, j$
from $x$ leaves $3m-6$ letters, which must be chosen from
$\lbrace c, f, i \rbrace$.  But there are only $m-2$ possible 
occurrences of each of the letters $c, f, i$ in $y$, so each occurrence of
these letters must appear on the diagonal of $y$ to get $x$.
Then these letters must be separated by $3m+1$ $0$'s.  
Thus $x = a b c^{m-2} d e f^{m-2} g h i^{m-2} j$.

Now an easy argument from the pumping lemma
shows that $L'$ is not context-free.  Hence $\diag(L)$ is not
context-free.
\end{proof}

\end{document}